\documentclass[aps,pra,10pt,twocolumn,superscriptaddress,floatfix,nofootinbib,showpacs,longbibliography,tikz,border=6mm]{revtex4-2}

\usepackage{amsmath}
\usepackage{tikz}

\usepackage{nicematrix}
\usepackage[utf8]{inputenc}
\usepackage[T1]{fontenc}     
\usepackage[british]{babel}  
\usepackage[sc,osf]{mathpazo}
\usepackage{times}
\usepackage[table]{xcolor}
\usepackage[scaled=0.86]{berasans}  
\usepackage[colorlinks=true, citecolor=purple, urlcolor=blue]{hyperref}
\usepackage{comment}
\makeatletter
\newcommand{\setword}[2]{%
  \phantomsection
  #1\def\@currentlabel{\unexpanded{#1}}\label{#2}%
}
\makeatother
\usepackage{comment}
\usepackage{graphicx} 
\usepackage[babel]{microtype}  
\usepackage{amsmath,amssymb,amsthm,bm,amsfonts,mathrsfs,bbm} 

\usepackage{xspace}  
\usepackage{pgf,tikz}
\usepackage{xcolor}
\usepackage{multirow}
\usepackage{array}
\usepackage{bigstrut}
\usepackage{braket}
\usepackage{color}
\usepackage{natbib}
\usepackage{multirow}
\usepackage{mathtools}
\usepackage{float}
\usepackage[caption = false]{subfig}
\usepackage{xcolor,colortbl}
\usepackage{color}

\newcommand{\be}{\begin{equation}}
\newcommand{\ee}{\end{equation}}
\newcommand{\ba}{\begin{eqnarray}}
\newcommand{\ea}{\end{eqnarray}}

\newcommand{\tr}{\operatorname{Tr}}

\newtheorem{theorem}{Theorem}

\newtheorem{definition}{Definition}
\newtheorem{proposition}{Proposition}

\def\>{\rangle}
\def\<{\langle}

\usepackage{centernot}
\usepackage{subfig}
\usepackage{filecontents}

\providecommand{\ket}[1]{| #1{\rangle}}
\providecommand{\bra}[1]{\langle #1|}

\usepackage[table]{xcolor}

\tikzset{
  box/.style = {rectangle, draw, minimum width=2.4cm, minimum height=6mm, align=center, font=\small},
  meas/.style = {rectangle, draw, minimum width=3.6cm, minimum height=7mm, align=center, font=\small, fill=blue!10},
  even/.style = {rectangle, draw, minimum width=2.4cm, minimum height=6mm, align=center, font=\small,fill=yellow!30},
  odd/.style = {rectangle, draw, minimum width=2.4cm, minimum height=6mm, align=center, font=\small,fill=red!30},
  arr/.style  = {-Stealth, thick}
}

\begin{document}
\title{On the Origin of Beyond-Classical Advantage in the Parity-Permutation Problem}

\author{Jayashree Karmakar}
\affiliation{S. N. Bose National Centre for Basic Sciences, Block JD, Sector III, Salt Lake, Kolkata 700 106, India}

\author{Biswadeep Chatterjee}
\affiliation{S. N. Bose National Centre for Basic Sciences, Block JD, Sector III, Salt Lake, Kolkata 700 106, India}

\author{Rafiuddin Gazi}
\affiliation{S. N. Bose National Centre for Basic Sciences, Block JD, Sector III, Salt Lake, Kolkata 700 106, India}

\author{Ananya Chakraborty}
\affiliation{S. N. Bose National Centre for Basic Sciences, Block JD, Sector III, Salt Lake, Kolkata 700 106, India}

\author{Snehasish Roy Chowdhury}
\affiliation{Physics and Applied Mathematics Unit, 203 B.T. Road Indian Statistical Institute Kolkata, 700 108, India}

\author{Manik Banik}
\affiliation{S. N. Bose National Centre for Basic Sciences, Block JD, Sector III, Salt Lake, Kolkata 700 106, India}

\author{Tamal Guha} 
\affiliation{Mathematical Institute, Slovak Academy of Sciences, Bratislava, Slovakia.}

\author{Sahil Gopalkrishna Naik}
\affiliation{S. N. Bose National Centre for Basic Sciences, Block JD, Sector III, Salt Lake, Kolkata 700 106, India}

\author{Kunika Agarwal}
\affiliation{S. N. Bose National Centre for Basic Sciences, Block JD, Sector III, Salt Lake, Kolkata 700 106, India}

\begin{abstract}
We investigate the task of identifying the parity (odd vs even) of an unknown permutation applied to $n$ particles. Classically, using fewer than $n$ distinct labels per particle limits the success probability to random guessing, whereas quantum mechanics, exploiting entanglement in both preparation and measurement, accomplishes the task perfectly with as few as $\big\lceil \sqrt{n}\big\rceil$ levels per particle [\href{https://doi.org/10.1103/yhyv-xnwq}{PRL {\bf 135}, 260603 (2025)}]. We show that even without entangled preparation, quantum theory still offers a probabilistic advantage over classical strategies. Moreover, such product preparations yield perfect success in locally quantum theories, where elementary systems are quantum but their composition follows the minimal tensor product structure of generalized probabilistic theories (GPTs). We further identify GPT models that accomplish the task with certainty without requiring entanglement either at the preparation stage or at the measurement stage.  Our central result establishes that the linear dimension of the elementary systems, rather than entanglement, is the fundamental resource governing the existence of probabilistic advantage in the permutation parity problem. In particular, below the required dimension threshold, no amount of entanglement can improve upon the random-guessing limit.

\end{abstract}


\maketitle

\section{Introduction} 
A central goal in quantum information theory
is to identify tasks where quantum systems provide advantages over the classical resources \cite{Dowling2003,Deutsch2020,Aspect2023}. Notable examples, that exhibit quantum advantage include Shor's algorithm \cite{Shor1997}, communication complexity reductions \cite{Buhrman2001,Buhrman2010,Zhong2021}, improved learning efficiency  \cite{Huang2022,Hinsche2023}, and query savings in oracle problems \cite{Simon1997,Bernstein1997}. Identifying the exact nonclassical feature(s) underlying a particular quantum advantage is essential both for a deeper theoretical understanding of its origin and for guiding experimental efforts to harness it in practice. For instance, in the circuit model of computation, the Gottesman–Knill theorem shows that stabilizer circuits—comprising of Clifford gates with computational basis state preparation and computational basis measurement—can be efficiently simulated on a classical computer \cite{Gottesman1998}. This identifies the non-Clifford $\mathrm{T}$-gate (or equivalently, the $\mathrm{T}$-state) as the essential resource for quantum advantage \cite{Bravyi2005,Knill2005,Campbell2012,Howard2017}. An analogous theorem in communication complexity settings identifies the non-Clifford resources necessary for quantum advantage in one-way protocols \cite{Chowdhury2026}.

In a recent work \cite{Diebra2025}, Diebra et al. identify an elegant problem demonstrating a striking quantum advantage. The task, called the $n$-permutation parity problem  which we denote as $\mathbb{P}[n]$, involves $n$ particles subject to a hidden permutation, with the binary question: is the permutation even or odd? Notably, the problem requires neither any oracle access nor any complex communication setup, and unlike the references \cite{Korff2004,Hayashi2005} the quantum advantage is obtained in non-asymptotic scenario. Classically, if each particle carries a distinct label (e.g., $n$ different colors), the permutation and its parity can be determined exactly. However, with even a single label missing, perfect parity identification becomes impossible: every permutation has an opposite-parity counterpart yielding the same labeling, restricting classical strategies to random guessing with success probability $\mathcal{P} = 1/2$. Quantum systems, by contrast, can achieve perfect success with $n$ particles each having local dimension as few as $\big\lceil\sqrt{n}\big\rceil$. It is important to note that entanglement is pivotal to the quantum advantage. Specifically, it is essential at two stages: first, in preparing the quantum systems in a suitable entangled state, and second, in performing a joint entangled-basis measurement that together enable perfect success.

In this work, we examine whether entanglement is strictly necessary—both at the stage of state preparation and at the stage of measurement—for achieving an advantage in the permutation parity problem. Quantum theory provides a quadratic advantage in solving the $\mathbb{P}[n]$ task, achieving perfect success with lower-dimensional systems than classically possible, where dimension is defined operationally as the number of perfectly distinguishable states. Note that, any quantum strategy achieving success above random guessing (i.e. $1/2$) in $\mathbb{P}[n]$ task with quantum systems of operational dimension strictly less than $n$ demonstrates an advantage.\smallskip
We begin by analyzing the minimal nontrivial instance, the $\mathbb{P}[3]$ game. While three classical two-level systems offer no advantage beyond random guessing, we show that three qubits prepared in a product state can achieve an optimal success probability of $7/8$. 
In the strategy, although the preparations involve only product states, the measurement necessarily requires entanglement. Allowing more general correlations at the level of states leads to a strict improvement: biseparable states achieve a higher success probability, revealing a hierarchy of resources within quantum theory.
These observations suggest that entanglement is not the essential resource underlying the advantage. Indeed, we show that an advantage in the permutation parity task is possible if and only if the linear dimension of the underlying state space is at least $n$, thereby identifying linear dimension—rather than entanglement—as the fundamental resource governing performance in this task.
To probe whether entanglement is fundamentally required, we move beyond the standard tensor product structure. Within locally quantum theory, where individual systems are described quantum mechanically but their composition is not quantum in general—a line of inquiry with deep roots in the study of quantum foundations \cite{Klay1987,Coecke2000,Barnum2010,Acin2010,delaTorre2012,Barnum2014,Naik2022,Lobo2022,Sen2022,Patra2023,Arai2024}-we show that the $\mathbb{P}[3]$ game can be accomplished perfectly using only product state preparations. Building on this, we further embed the problem in a more extended setting of generalized probabilistic theories (GPTs) \cite{Plvala2023}, which provides a unifying mathematical structure for exploring physical theories beyond quantum mechanics. We construct explicit instances of GPTs where the permutation parity problem can be solved perfectly without employing entanglement either in preparation or in measurement. 
These observations suggest that entanglement is not the essential resource underlying the advantage. Indeed, we establish a general necessary and sufficient condition for achieving success above random guessing in the $\mathbb{P}[n]$ task: such an advantage is possible if and only if the linear dimension of the underlying state space is at least $n$. This identifies linear dimension, rather than entanglement, as the fundamental resource governing performance in the permutation parity problem.
\vspace{-0.1cm}
\section{Playing the $\mathbb{P}[3]$ game in Quantum Theory} 
Classically, the permutation parity problem $\mathbb{P}[n]$ cannot be solved with a success probability exceeding that of random guessing unless $n$ distinct labels are used. Quantum mechanics, however, achieves perfect success with only $\lceil \sqrt{n} \rceil$ distinguishable states per particle by leveraging entanglement. For instance, $\mathbb{P}[3]$ is solved deterministically using the three-qubit W-state $\ket{W}=(\ket{011}+\Gamma\ket{101}+\Gamma^{2}\ket{110})/\sqrt{3}$, where $\Gamma=\exp(2\pi i/3)$, as its even- and odd-parity orbits are orthogonal. However, a quantum advantage need not be deterministic, any strategy that surpasses the classical success probability using quantum systems of an operational dimension less than $n$ demonstrates a probabilistic advantage. We start with establishing such an advantage for $\mathbb{P}[3]$ without entangled preparation, that is, a fully product preparation, also providing an upper bound on the success probability achievable with product states and then showing such an advantage for bi-separable preparation.\\

\subsection{Fully Product Preparation}
We start with an explicit construct of a fully product 3-qubit state. The following proposition establishes the optimal success probability achievable in the $\mathbb{P}[3]$ task under this restriction.
\begin{proposition}\label{prop1}
The maximum success probability for solving the $\mathbb{P}[3]$ task with three-qubit product states is $\frac{7}{8}$.
\end{proposition}
\begin{proof}
For the three-qubit Hilbert state $(\mathbb{C}^2)^{\otimes 3}$ we consider the following orthonormal basis:
\begin{align}
\left.\begin{aligned}
\ket{\psi_1}&=\ket{000},\hspace{1.5cm} \ket{\psi_2}=\ket{111},\\
\ket{\psi_3}&=\tfrac{1}{\sqrt{3}}(\ket{001}+\ket{010}+\ket{100}),\\
\ket{\psi_4}&=\tfrac{1}{\sqrt{3}}(\ket{110}+\ket{101}+\ket{011}),\\
\ket{\psi_5}&=\tfrac{1}{\sqrt{3}}(\ket{001}+\Gamma\ket{010}+\Gamma^2\ket{100}),\\
\ket{\psi_6}&=\tfrac{1}{\sqrt{3}}(\ket{110}+\Gamma\ket{101}+\Gamma^2\ket{011}),\\
\ket{\psi_7}&=\tfrac{1}{\sqrt{3}}(\ket{001}+\Gamma^2\ket{010}+\Gamma\ket{100}),\\
\ket{\psi_8}&=\tfrac{1}{\sqrt{3}}(\ket{110}+\Gamma^2\ket{101}+\Gamma\ket{011})
\end{aligned}\right\}
\end{align}

The states $\{\ket{\psi_i}\}_{i=1}^4$ span the symmetric subspace and are invariant under any permutation. In contrast, the subspaces spanned by $\{\ket{\psi_5}, \ket{\psi_6}\}$ and $\{\ket{\psi_7}, \ket{\psi_8}\}$ are invariant under even permutations but are mapped into one another under odd permutations. Table~\ref{tab1} summarizes the transformation of the basis states under even and odd permutations.  
\begin{table}[h!]
\begin{tabular}{|c|c|c||c|c|c|}
\hline
$123~(\mathrm{E}_1)$  & $312~(\mathrm{E}_2)$ & $213~(\mathrm{E}_3)$ & $213~(\mathrm{O}_1)$ & $132~(\mathrm{O}_2)$ & $321~(\mathrm{O}_3)$ \\ \hline\hline
$\ket{\psi_1}$ & $\ket{\psi_1}$ & $\ket{\psi_1}$ & $\ket{\psi_1}$ & $\ket{\psi_1}$ & $\ket{\psi_1}$ \\ \hline
$\ket{\psi_2}$ & $\ket{\psi_2}$ & $\ket{\psi_2}$ & $\ket{\psi_2}$ & $\ket{\psi_2}$ & $\ket{\psi_2}$ \\ \hline
$\ket{\psi_3}$ & $\ket{\psi_3}$ & $\ket{\psi_3}$ & $\ket{\psi_3}$ & $\ket{\psi_3}$ & $\ket{\psi_3}$ \\ \hline
$\ket{\psi_4}$ & $\ket{\psi_4}$ & $\ket{\psi_4}$ & $\ket{\psi_4}$ & $\ket{\psi_4}$ & $\ket{\psi_4}$ \\ \hline\hline
$\ket{\psi_5}$ & $\Gamma\ket{\psi_5}$   & $\Gamma^2\ket{\psi_5}$ & $\ket{\psi_7}$ & $\Gamma\ket{\psi_7}$   & $\Gamma^2\ket{\psi_7}$ \\ \hline
$\ket{\psi_6}$ & $\Gamma\ket{\psi_6}$   & $\Gamma^2\ket{\psi_6}$ & $\ket{\psi_8}$ & $\Gamma\ket{\psi_8}$   & $\Gamma^2\ket{\psi_8}$ \\ \hline
$\ket{\psi_7}$ & $\Gamma^2\ket{\psi_7}$ & $\Gamma\ket{\psi_7}$   & $\ket{\psi_5}$ & $\Gamma^2\ket{\psi_5}$ & $\Gamma\ket{\psi_5}$   \\ \hline
$\ket{\psi_8}$ & $\Gamma^2\ket{\psi_8}$ & $\Gamma\ket{\psi_8}$   & $\ket{\psi_6}$ & $\Gamma^2\ket{\psi_6}$ & $\Gamma\ket{\psi_6}$   \\ \hline
\end{tabular}
\caption{Behavior of the basis states $\{\ket{\psi_i}\}_{i=1}^8$ under even and odd permutations.}\label{tab1}
\end{table}

\noindent With an arbitrary three qubit state $\ket{\psi}=\sum_{i=1}^8a_i\ket{\psi_i}$, the question of identifying the odd and even parity boils down to discriminating the following two mixed states
\begin{align}
\chi^{\mathrm{E}}:=\frac{1}{3}\sum_{\pi\in E_3}\pi(\ket{\psi}\bra{\psi}),~~\chi^{\mathrm{O}}:=\frac{1}{3}\sum_{\pi\in O_3} \pi(\ket{\psi}\bra{\psi}).
\end{align}
where $E_n, O_n \subset S_n$ denote the sets of even and odd permutations respectively and $S_n$ is  the symmetric group of $n$ elements.

The states $\chi^{\mathrm{E}}$ and $\chi^{\mathrm{O}}$ when represented in $\{\ket{\psi_i}\}_{i=1}^8$ basis take the block-diagonal forms
\begin{align}
\chi^{\mathrm{E}}= \renewcommand\arraystretch{1.5}
\left[\begin{array}{@{}c|c@{}c@{}}
\mathbb{M} & \mathbb{O}_4& \\
\hline
\mathbb{O}_4 & \begin{array}{@{}c|c@{}c@{}}
\mathbb{A} & \mathbb{O}_2& \\
\hline
\mathbb{O}_2 & \mathbb{B}&
\end{array}&
\end{array}\right],~~~~
\chi^{\mathrm{O}}=\renewcommand\arraystretch{1.5}
\left[\begin{array}{@{}c|c@{}c@{}}
\mathbb{M} & \mathbb{O}_4& \\
\hline
\mathbb{O}_4 & \begin{array}{@{}c|c@{}c@{}}
\mathbb{B} & \mathbb{O}_2& \\
\hline
\mathbb{O}_2 & \mathbb{A}&
\end{array}&
\end{array}\right],
\end{align}
where $\mathbb{M}$ is a $4\times 4$ matrix and $\mathbb{O}_k$ represents $k\times k$ zero matrix, and 
\begin{align}
\mathbb{A}=\begin{bmatrix}
|a_5|^2 & a_{5}a_{6}^*\\
a_{5}^*a_{6} & |a_6|^2
\end{bmatrix},~~~
\mathbb{B}=\begin{bmatrix}
|a_7|^2& a_{7}a_{8}^* \\
a_{7}^*a_{8}&|a_8|^2
\end{bmatrix}.
\end{align}
The success probability $\mathcal{P}^{[3]}_{\psi}$ of the task $\mathbb{P}[3]$ with the initial prepared state $\ket{\psi}\in(\mathbb{C}^2)^{\otimes3}$ thus boils down to the success of discriminating the states $\chi^{\mathrm{E}}$ and $\chi^{\mathrm{O}}$, which according to the seminal result of Helstrom \cite{Helstrom1969} reads as
\begin{align*}
\mathcal{P}^{[3]}_{\psi}=\frac{1}{2}\left(1+\tfrac{1}{2}||\chi^{\mathrm{E}}-\chi^{\mathrm{O}}||_{\text{tr}}\right)=\frac{1}{2}\left(1+||\mathbb{A}-\mathbb{B}||_{\text{tr}}\right).
\end{align*}
Denoting eigenvalues of the operator $\mathbb{A}-\mathbb{B}$ as $\lambda_1$ and $\lambda_2$, we have
\begin{align}
\mathcal{P}^{[3]}_{\psi}&=\frac{1}{2}\Big[1+|\lambda_{1}|+|\lambda_{2}|\Big]=\frac{1}{2}\Big[1+\Big(4|a_5a_8-a_6a_7|^2\nonumber\\
&\hspace{1.5cm}+\big(|a_5|^2+|a_6|^2-|a_7|^2-|a_8|^2\big)^2\Big)^{\frac{1}{2}}\Big].\label{optsucc}
\end{align}
Without loss of any generality, we can consider a three qubit product of the form
{\footnotesize
\begin{align}
\ket{\psi_p}&=\ket{0}\otimes\left(\mathrm{C}_{\theta_1}\ket{0}+\mathrm{S}_{\theta_1}\ket{1}\right)\otimes\left(\mathrm{C}_{\theta_2}\ket{0}+e^{i\eta}\mathrm{S}_{\theta_2}\ket{1}\right),
\end{align}}
where, $\mathrm{C}_{\alpha}\equiv\cos(\alpha/2)$ and $\mathrm{S}_{\alpha}\equiv\sin(\alpha/2)$.
Expressing $\ket{\psi_p}=\sum_{i=1}^8a_i\ket{\psi_i}$, we obtain the coefficients  
\begin{align}
\left.\begin{aligned}
a_5&=\frac{1}{\sqrt{3}}(e^{i\eta}C_{\theta_{1}}S_{\theta_{2}}+\Gamma^{2}S_{\theta_{1}}C_{\theta_{2}})\\
a_6&=\frac{\Gamma}{\sqrt{3}}e^{i\eta}S_{\theta_{1}}S_{\theta_{2}}\\
a_7&=\frac{1}{\sqrt{3}}(e^{i\eta}C_{\theta_{1}}S_{\theta_{2}}+\Gamma S_{\theta_{1}}C_{\theta_{2}})\\
a_8&=\frac{\Gamma^2}{\sqrt{3}}e^{i\eta}S_{\theta_{1}}S_{\theta_{2}}
\end{aligned}\right\}.
\end{align}
Now, putting these in Eq.(\ref{optsucc}) and optimizing over  $\theta_1,~\theta_2,~\&~\eta$ we have 
\begin{align}
\left.\begin{aligned}
&\mathcal{P}^{[3]}_{\psi^\star_{p}}=\frac{1}{2}\left(1+\frac{3}{4}\right)=\frac{7}{8},~\text{achieved~at} \\
&\theta_1=\frac{2\pi}{3},~~\theta_2=\frac{4\pi}{3},~~\eta=0.
\end{aligned}\right\}.
\end{align}
Accordingly the product state reads as

(see Fig.~\ref{fig1}):
\begin{align}
\ket{\psi^\star_p}=\ket{0}\otimes\left(\tfrac{1}{2}\ket{0}+\tfrac{\sqrt{3}}{2}\ket{1}\right)\otimes\left(-\tfrac{1}{2}\ket{0}+\tfrac{\sqrt{3}}{2}\ket{1}\right).\hspace{-.1cm}  
\end{align}
This completes the proof. 
\end{proof}
Having analyzed the scenario under product state preparation, we now seek to determine the nature of the measurement that achieves optimal success for this optimal product state preparation.The task described above is essentially that of discriminating between two quantum states. It is well known that the optimal measurement for distinguishing between two given states is the Helstrom measurement.

For $\rho, \sigma \in \mathbb{C}^d$, the Helstrom measurement is unique if and only if $\mathrm{rank}(\rho - \sigma) = d$. 

If we write down $\chi_\mathrm{E}$ and $\chi_\mathrm{O}$ using the state in Eq.~(9), we observe that the operator $\chi_\mathrm{E} - \chi_\mathrm{O}$ has eigenvalues $\{0, 0, 0, 0, \tfrac{3}{8}, \tfrac{3}{8}, -\tfrac{3}{8}, -\tfrac{3}{8}\}$, with the corresponding eigenvectors:
\begin{align*}
\ket{\phi_i}=\ket{\psi_i}\quad\text{for}~ i=1,2,3,4\quad\text{and},\hspace{1cm}\\
\big(\ket{\phi_5},\ket{\phi_6},\ket{\phi_7},\ket{\phi_8}\big)^{\mathrm{T}}=\mathbb{F}\big(\ket{\psi_5},\ket{\psi_6},\ket{\psi_7},\ket{\psi_8}\big)^{\mathrm{T}}.  
\end{align*}
where,
\begin{align*}
\mathbb{F} =
\begin{bmatrix}
\frac{\Gamma^2}{2} & \frac{\Gamma-1}{2\sqrt{3}} & \frac{\Gamma}{2} & \frac{\Gamma^2-1}{2\sqrt{3}} \\
\frac{\Gamma-1}{2\sqrt{3}} & \frac{\Gamma^2}{2} & \frac{\Gamma^2-1}{2\sqrt{3}} & \frac{\Gamma}{2} \\
\frac{\Gamma^2}{2} & \frac{1-\Gamma}{2\sqrt{3}} & \frac{\Gamma}{2} & \frac{1-\Gamma^2}{2\sqrt{3}} \\
\frac{1-\Gamma}{2\sqrt{3}} & \frac{\Gamma^2}{2} & \frac{1-\Gamma^2}{2\sqrt{3}} & \frac{\Gamma}{2} 
\end{bmatrix}.
\end{align*}
Consequently, the Helstrom measurement $\mathrm{M}\equiv\big\{\mathrm{h_E},\mathrm{h_O}\big\}$ that optimally distinguish the states $\chi^{\mathrm{E}}$ and $\chi^{\mathrm{O}}$ has different choices. For instance, in Table~\ref{tab2} we list three such measurements that achieve the optimal success $7/8$ in Proposition~\ref{prop1}. While both effects in measurement $\mathrm{M}^3$ are NPT across every bi-separable cut and hence are entangled, in $\mathrm{M}^2$, one effect is NPT and the other is PPT across all bipartitions. On the other-hand, in $\mathrm{M}^1$, both effects are PPT across all bi-separable cuts. However, we are currently unable to determine the entanglement properties of these effects. These questions remain open to be resolved.

\begin{table}[t!]
\begin{tabular}{|c|c|c|c|}
\hline
Measurement & Components & Rank & Feature \\ \hline\hline
$\mathrm{M}^1$ & 
$\begin{aligned}
\mathrm{h^1_E} &=\Phi_{56}+\frac{1}{2}\Phi_{1234} \\
\mathrm{h^1_O} &=\Phi_{78}+\frac{1}{2}\Phi_{1234}
\end{aligned}$ 
& 
$\begin{aligned}
&6 \\
&6
\end{aligned}$ 
& 
$\begin{aligned}
\text{PPT} \\
\text{PPT}
\end{aligned}$ \\ \hline\hline
$\mathrm{M}^2$ & 
$\begin{aligned}
\mathrm{h^2_E} &= \Phi_{123456} \\
\mathrm{h^2_O} &= \Phi_{78}
\end{aligned}$ 
& 
$\begin{aligned}
&6 \\
&2
\end{aligned}$ 
& 
$\begin{aligned}
\text{PPT} \\
\text{NPT}
\end{aligned}$ \\ \hline\hline

$\mathrm{M}^3$ & 
$\begin{aligned}
\mathrm{h^3_O} &= \Phi_{1256} \\
\mathrm{h^3_O} &= \Phi_{3478}
\end{aligned}$ 
& 
$\begin{aligned}
&4 \\
&4
\end{aligned}$ 
& 
$\begin{aligned}
\text{NPT} \\
\text{NPT}
\end{aligned}$ \\ \hline
\end{tabular}
\caption{All these measurements optimally distinguishes the states $\chi_{\mathrm{E}}$ and $\chi_{\mathrm{O}}$. Here,  $\Phi_{ij}:=\ket{\phi_i}\bra{\phi_i}+\ket{\phi_j}\bra{\phi_j}$ and so on.}\label{tab2}
\vspace{-.5cm}
\end{table}

A natural question that arises at this point is: what is the success probability of the $\mathbb{P}[3]$ game when the shared state is only biseparable. In our next Proposition we address this question.\\

\subsection{Bi-separable Preparation}
We next consider the case where the shared state is biseparable, i.e., it belongs to the convex hull of states that are separable with respect to one of the three bipartitions. The following proposition establishes the optimal success probability achievable in this setting.
\begin{proposition}\label{prop2}
The maximum success probability for solving the $\mathbb{P}[3]$ task with three-qubit biseparable states is  $\frac{1}{2}\big(1+\sqrt{2/3}\big)$.
\end{proposition}
\begin{proof}
Without loss of any generality, consider a three qubit biseparable state as 
\begin{align}
\left.\begin{aligned}
&\ket{\psi_{bi}}=\ket{0}_{A}\otimes\left(C_{\alpha}\ket{\phi}\ket{\xi}+S_{\alpha}\ket{\phi^{\perp}}\ket{\xi^{\perp}}\right)_{BC},\\    
&\ket{\phi}=C_{\theta_{1}}\ket{0}+S_{\theta_{1}}\ket{1},
\ket{\xi}=C_{\theta_{2}}\ket{0}+S_{\theta_{2}}e^{i\eta}\ket{1}
\end{aligned}\right\}\hspace{-.2cm}
\end{align}
where $\ket{\phi^\perp}$ denotes the ket orthogonal to $\ket{\phi}$. Expressing $\ket{\psi_{bi}}=\sum_{i=1}^8a_i\ket{\psi_i}$, we obtain the coefficients  
\begin{align}
\left.\begin{aligned}
a_5=\frac{1}{\sqrt{3}}&\Big[e^{i\eta}(C_{\alpha}C_{\theta_{1}}S_{\theta_{2}}-S_{\alpha}S_{\theta_{1}}C_{\theta_{2}})\\
&\hspace{1cm}+\Gamma^2(C_{\alpha}S_{\theta_{1}}C_{\theta_{2}}-S_{\alpha}C_{\theta_{1}}S_{\theta_{2}})\Big],\\
a_6=\frac{\Gamma}{\sqrt{3}}&e^{i\eta}(C_{\alpha}S_{\theta_{1}}S_{\theta_{2}}+S_{\alpha}C_{\theta_{1}}C_{\theta_{2}}),\\
a_7=\frac{1}{\sqrt{3}}&\Big[e^{i\eta}(C_{\alpha}C_{\theta_{1}}S_{\theta_2}-S_{\alpha}S_{\theta_{1}}C_{\theta_{2}})\\
&\hspace{1cm}+\Gamma(C_{\alpha}S_{\theta_{1}}C_{\theta_{2}}-S_{\alpha}C_{\theta_{1}}S_{\theta_{2}})\Big],\\
a_8=\frac{\Gamma^2}{\sqrt{3}}&e^{i\eta}(C_{\alpha}S_{\theta_{1}}S_{\theta_{2}}+S_{\alpha}C_{\theta_{1}}C_{\theta_{2}})
\end{aligned}\right\}.
\end{align}
Now, putting these in Eq.(\ref{optsucc}) and optimizing over  $\alpha~,\theta_1,~\theta_2,~\&~\eta$ we have 
\begin{align}
\left.\begin{aligned}
&\mathcal{P}^{[3]}_{\psi^\star_{bi}}=\frac{1}{2}\left(1+\sqrt{\frac{2}{3}}\right),~\text{achieved~at} \\
&\alpha=\frac{\pi}{12},~~\theta_1=\frac{15\pi}{23},~~\theta_2=\frac{8\pi}{23},~~\eta=0.
\end{aligned}\right\}
\end{align}
Accordingly the biseparable reads as
\begin{align}
\ket{\psi^\star_{bi}}=\ket{0}\otimes\left(\tfrac{\sqrt{3}+1}{2\sqrt{2}}\ket{\phi\xi}+\tfrac{\sqrt{3}-1}{2\sqrt{2}}\ket{\phi^{\perp}\xi^{\perp}}\right),\label{biopt}
\end{align}
with, 
\begin{align*}
   \ket{\phi}&=C_{\frac{8\pi}{23}}\ket{0}+S_{\frac{8\pi}{23}}\ket{1} \\
   \ket{\xi}&=C_{\frac{15\pi}{23}}\ket{0}+S_{\frac{15\pi}{23}}\ket{1}
\end{align*}
This completes the proof. 
\end{proof}

Having analyzed the scenario under bi-separable state preparation, we now seek to determine the nature of the measurement that achieves optimal success for this optimal Bi-separable state preparation.\\
By following the same procedure as in product state preparation—specifically, by computing the eigenvalues and eigenvectors of $\chi_E - \chi_O$ with the state in Eq.\textcolor{red}{(16)}—we observe that all measurement settings presented in Table\textcolor{red}{II} yield the optimal bi-separable success probability.

The entangled state in Eq.(\ref{biopt}) is very special. It appears in Mean-King problem \cite{Vaidman1987} as well as in Elegant Joint measurement \cite{Gisin2019} (see also \cite{Patra2026, Agarwal2026}). While Proposition \ref{prop2} establishes that perfect success of $\mathbb{P}[3]$ task with three-qubit system requires genuine entangled state preparation as well as entangled measurement, in the next section we analyze this task within more general framework of GPT.     

\section{Parity games in GPT\MakeLowercase{s}}
The origin of GPTs dates back to the 1960s, with the motivation of axiomatic derivation of deriving Hilbert space quantum mechanics \cite{Mackey1963,Ludwig1967,Ludwig1968,Davies1970}. With the advent of quantum information theory, the framework has gained renewed prominence as a versatile tool for probing the foundations of quantum mechanics and exploring possible alternatives \cite{Hardy2001,Barrett2007,Barnum2011,Masanes2011,Chiribella2011}. While GPTs specify the structure of elementary systems, they also prescribe rules for composing them into multipartite systems.

\subsection{Framework of GPTs} 
\noindent An elementary system $\mathcal{S}$ is specified by the triple $(\Omega_{\mathcal{S}}, \mathcal{E}_{\mathcal{S}}, \mathcal{T}_{\mathcal{S}})$ of state space, effect space, and the set of transformations. The state space $\Omega_{\mathcal{S}}$ consists of normalized states, assumed to be a convex and compact subset of a finite-dimensional real vector space $\mathbb{R}^d$. A convenient representation of a state is given by $\omega = (\vec{v},\,1)^{\mathrm{T}}$, where $\vec{v}\in \mathbb{R}^{d-1}$, and the final component is fixed to encode normalization. The effect space $\mathcal{E}_{\mathcal{S}}$ comprises all linear functionals mapping normalized states to probabilities, i.e. $\mathcal{E}_{\mathcal{S}} \equiv \{e:\Omega_{\mathcal{S}}\mapsto [0,1]\}$. Two special effects are the unit effect $u$ and the zero effect $\Theta$ respectively satisfying $u(\omega)=1$ and $\Theta(\omega)=0$ for all $\omega \in \Omega_{\mathcal{S}}$. It is often convenient to consider unnormalized states and effects, which respectively form convex cones: $\mathbf{V}_{\Omega} := \{\mu\,\omega \;|\; \omega\in\Omega_{\mathcal{S}},~\mu\in\mathbb{R}_{\ge 0}\},~~\mathbf{V}_{\mathcal{E}} := \{\mu\,e \;|\; e\in\mathcal{E}_{\mathcal{S}},~\mu\in\mathbb{R}_{\ge 0}\}$, dual to each other. The transformation space $\mathcal{T}_{\mathcal{S}}$ contains linear maps that send states to states, i.e. $T:~\Omega_S\mapsto\Omega_S$.  
\begin{definition}\label{def1}
[Linear dimension] For a GPT system $\mathcal{S}$, its linear dimension $\mathtt{Ld}_{\mathcal{S}}$ corresponds to the dimension of the vector space spanned by its unnormalized states. 
\end{definition}
\noindent An alternative and more operational notion is the operational dimension (OD). A set of states $\{\omega_j\}\subseteq\Omega_{\mathcal{S}}$ are said to be perfectly distinguishable if there exists a measurement $\mathrm{M}\equiv\{e_i\mid \sum_i e_i = u\}$ such that $e_i(\omega_j) = \delta_{ij}$.
\begin{definition}\label{def2}
[Operational dimension] For a GPT system $\mathcal{S}$, its operational dimension $\mathtt{Od}_{\mathcal{S}}$ corresponds to the maximum number of states that can be perfectly distinguished by a single measurement. 
\end{definition}
\noindent For two GPT systems OD can be same even when their LDs are different. For instance, a classical bit and a qubit have LDs $2$ and $4$ respectively, while both have OD $2$. In fact, a GPT system, such as the hyper-sphere theory has OD$=2$, even though its LD can be arbitrarily large \cite{Masanes2014}. As we shall show, both LD and OD play a central role in analyzing the permutation parity problem. 

Given two systems $\mathcal{S}_A$ and $\mathcal{S}_B$, the GPT framework also prescribes rules for constructing the corresponding composite system. Any such composition is required to satisfy the no-signaling (NS) principle, which forbids instantaneous communication between subsystems. A further assumption, known as local tomography (LT) \cite{Hardy2011}, requires that a composite state can be uniquely characterized (to arbitrary precision) using only local operations and classical communication (LOCC). In addition, the no-restriction hypothesis \cite{Chiribella2010} is often imposed, which demands that every mathematically well-defined state or effect consistent with valid probability assignments be regarded as physically admissible. Any composition obeying both NS and LT lies between two extreme constructions: the minimal tensor product and the maximal tensor product \cite{Namioka1969,Barker1976,Aubrun2021}.
\begin{definition}\label{def3}
The state space of the minimal tensor product $\mathcal{S}_A \otimes_{\min} \mathcal{S}_B$ is the convex hull of all product states, i.e. $\Omega^{\min}_{AB} := \mathrm{Conv.hull}\{\;\omega_A \otimes \omega_B \;\mid\; \omega_A \in \Omega_A \subsetneq \mathbb{R}^{d_A},~ \omega_B \in \Omega_B \subsetneq \mathbb{R}^{d_B}\;\}$. The effect space $\mathcal{E}^{\min}_{AB}$ consists of all linear functionals on $\mathbb{R}^{d_A} \otimes \mathbb{R}^{d_B}$ mapping $\Omega^{\min}_{AB}$ to $[0,1]$.
\end{definition}

Assuming local description to be quantum, the minimal tensor product corresponds to allowing only separable states in the state space (fully separable states for more than two subsystems). However, the associated effect space is larger than that of standard quantum theory, since entangled witness operators not being positive operator are admissible effects under minimal composition. By contrast, in the maximal tensor product, the situation is reversed: the effect space is constrained, while the state space is enlarged beyond the quantum set. For further details of the framework we refer to the recent review in \cite{Plvala2023}.
\begin{figure}[t!]
\centering
\includegraphics[width=.95\linewidth]{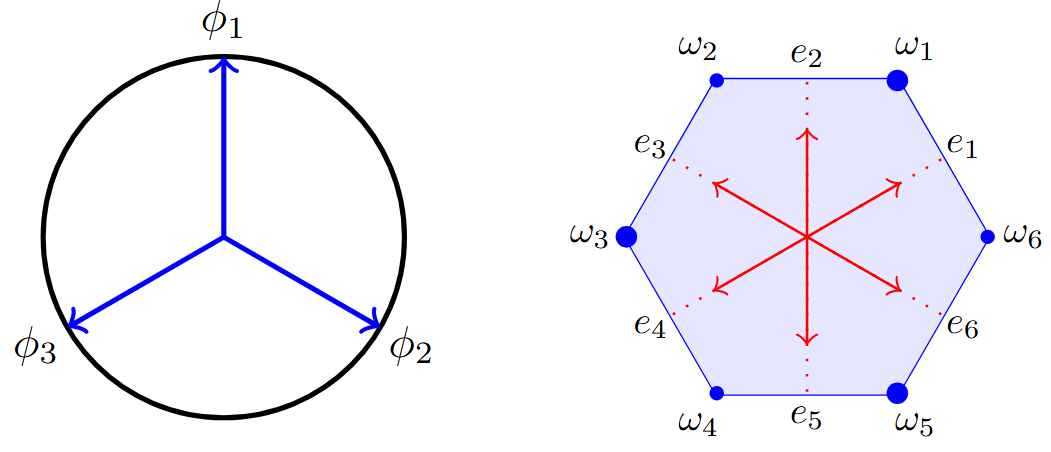}
\caption{Left: Three qubits prepared in the product state  $\ket{\phi_1}\otimes\ket{\phi_2}\otimes\ket{\phi_3}$, with $\ket{\phi_1} = \ket{0},
\ket{\phi_2} = \tfrac{1}{2}\ket{0}-\tfrac{\sqrt{3}}{2}\ket{1},
\ket{\phi_3} = \tfrac{1}{2}\ket{0}+\tfrac{\sqrt{3}}{2}\ket{1}$, yields a quantum success probability of $7/8$ in the $\mathbb{P}[3]$ task, while the same preparation achieves perfect success under minimal tensor product composition. Right: In the hexagon model, three systems prepared in the product state $\omega_1\otimes\omega_3\otimes\omega_5$ accomplish $\mathbb{P}[3]$ with certainty, without requiring entanglement either at the preparation or at the measurement stage.}\label{fig1}
\vspace{-.2cm}
\end{figure}

\subsection{$\mathbb{P}[3]$ in Locally Quantum Theory}
At this stage, a natural question is whether one can surpass the success probability of Proposition~\ref{prop2} by adopting minimal composition while retaining quantum descriptions of the local systems. In minimal composition, the effect space is strictly larger than in standard quantum theory (see Definition~\ref{def2}), which enhances the ability to distinguish separable states. Indeed, certain non-orthogonal separable states become perfectly distinguishable in this framework \cite{Arai2019}. As we show next, this feature carries significant consequences for the permutation parity problem.
\begin{proposition}\label{prop3}
In minimal composition of three qubits, the permutation parity task $\mathbb{P}[3]$ can be accomplished perfectly with product preparation. 
\end{proposition}
\begin{proof}
Any valid effect in minimal tensor product composition theory of three qubits is a operator in $\mathrm{Herm}((\mathbb{C}^2)^{\otimes 3})$ satisfying the requirement that on any fully product three qubit state it yields positive probability. To perfectly win the $\mathbb{P}[3]$ task with the product preparation given in Theorem \ref{prop1}, we thus need a two-outcome measurement $\mathrm{M}\equiv\{\mathrm{h}'_{\mathrm{E}},\mathrm{h}'_{\mathrm{O}}~|~\mathrm{h}'_{\mathrm{E}}+\mathrm{h}'_{\mathrm{O}}=\mathbf{I}^{\otimes 3}\}$ such that $\tr[\mathrm{h}'_{a}\chi^b]=\delta_{ab}$ with $a,b\in\{\mathrm{E},\mathrm{O}\}$.

\noindent We will be using the notation $\{\sigma_0:=\mathbf{I},\sigma_1:=\sigma_X,\sigma_2:=\sigma_Y,\sigma_3:=\sigma_Z\}$ for the Pauli operators and $\sigma_{\mu}\otimes\sigma_{\nu}\otimes\sigma_{\eta}\equiv \sigma_{\mu} \sigma_{\nu}\sigma_{\eta} ~\forall~ \mu,\nu,\eta \in \{0,1,2,3\}$. 
 
\noindent Considering $\mathrm{h}'_{\mathrm{E}}$ as general operator in $\mathcal{L}(\mathbb{C}^{2^{\otimes 3}})$,it can be expressed in Hilbert-Schmidt form as 
\begin{align}
\mathrm{h}'_{\mathrm{E}}&=\frac{1}{2}\Big(\alpha~\mathbf{I}^{\otimes 3}
+\sum_{i=1}^{3} \alpha_i^{(1)}\sigma_i\mathbf{I}\mathbf{I}+\sum_{i=1}^{3} \alpha_i^{(2)}\mathbf{I}\sigma_i\mathbf{I}+\sum_{i=1}^{3}\alpha_i^{(3)}\mathbf{I}\mathbf{I}\sigma_i\nonumber\\
&\hspace{.8cm}+\sum_{i,j=1}^{3}\beta_{ij}^{(1)}\mathbf{I}\sigma_i\sigma_j
+\sum_{i,j=1}^{3}\beta_{ij}^{(2)}\sigma_i\mathbf{I}\sigma_j+ \sum_{i,j=1}^{3} \beta_{ij}^{(3)}\sigma_i\sigma_j\mathbf{I}\nonumber\\
&\hspace{4.2cm}+\sum_{i,j,k=1}^{3} \gamma_{ijk}~\sigma_i\sigma_j\sigma_k\Big),\label{he}    
\end{align}
where the coefficients $\alpha,\alpha_{i}^{(j)},\beta_{ij}^{(k),\gamma_{ijk}} \in\mathbb{R}$ with $i,j,k \in [3]$. Accordingly, the operator $\mathrm{h}'_{\mathrm{O}}~(=\mathbf{I}^{\otimes3}-\mathrm{h}'_{\mathrm{E}})$ is defined. 

\noindent The party permutations in locally quantum theory is described by some global unitary operators which are allowed reversible operation in minimal tensor product composition since fully product states remain fully product under such operations. Notably, the states $\chi^{\mathrm{E}},~\chi^{\mathrm{O}}$ satisfy the following symmetry feature
\begin{align}
&U_{\mathrm{E}}\chi^{\mathrm{E}/\mathrm{O}}U^\dagger_{\mathrm{E}}=\chi^{\mathrm{E}/\mathrm{O}},\quad\quad U_{\mathrm{O}}\chi^{\mathrm{E}/\mathrm{O}}U^\dagger_{\mathrm{O}}=\chi^{\mathrm{O}/\mathrm{E}},
\end{align}    
where $U_{\mathrm{E}}$, $U_{\mathrm{O}}$ respectively denote unitaries corresponding to even and odd permutations. These symmetries further ensure to choose the operators $\mathrm{h}'_{\mathrm{E}}$ and $\mathrm{h}'_{\mathrm{O}}$ satisfying 
\begin{align}
&U_{\mathrm{E}}\mathrm{h}'_{\mathrm{E}/\mathrm{O}}U^\dagger_{\mathrm{E}}=\mathrm{h}'_{\mathrm{E}/\mathrm{O}},\quad\quad U_{\mathrm{O}}\mathrm{h}'_{\mathrm{E}/\mathrm{O}}U^\dagger_{\mathrm{O}}=\mathrm{h}'_{\mathrm{O}/\mathrm{E}},
\end{align}    
which substantially simplifies the operator's choice in Eq.(\ref{he}) by fixing 
\begin{subequations}
\begin{align}
&\alpha=1;~~\alpha^{(p)}_i=\beta^{(p)}_{ii}=0;\\
&\beta^{(p)}_{ij}=-\beta^{(p)}_{ji};~\beta^{(p)}_{ij}=\beta^{(q)}_{ij}~:~i\neq j~\&~k\neq l;\\
&\gamma_{iij}=\gamma_{iji}=\gamma_{jii}=\gamma_{iii}=0~:~i\neq j;\\
&\gamma_{ijk}=\gamma_{jkk}=\gamma_{kij}=-\gamma_{ikj}=-\gamma_{jik}=-\gamma_{kji}~:~\nonumber\\
&\hspace{3cm} i\neq j\neq k\neq i;
\end{align}
\end{subequations}
where all the indices take values $\{1,2,3\}$. Accordingly, the operator $\mathrm{h}_{\mathrm{E}}$ takes the form
\begin{align}
\mathrm{h}'_{\mathrm{E}}&=\tfrac{1}{2}\Big(\beta_{12}[[\sigma_1,\sigma_2,\mathbf{I}]]+\beta_{23}[[\sigma_2,\sigma_3,\mathbf{I}]]\nonumber\\
&\hspace{1cm}+\beta_{31}[[\sigma_3,\sigma_1,\mathbf{I}]]+\gamma[[\sigma_1,\sigma_2,\sigma_3]]+\mathbf{I}^{\otimes3}\Big),\label{hesim1}
\end{align}
where $[[u,v,w]]:=u\otimes v\otimes w+v\otimes w\otimes u+w\otimes u\otimes v-v\otimes u\otimes w-u\otimes w\otimes v-w\otimes v\otimes u$.

The particular choices of the states $\phi_1,~\phi_2,~\phi_3$ (see Fig.~\ref{fig1}) induces another interesting symmetry. Consider the unitary operation $U_{\mathrm{T}}$ that is $2\pi/3$ rotation about the axis joining origin and the vector $\hat{\mathrm{T}}:=\tfrac{1}{\sqrt{3}}(1,1,1)$. Now, we have
\begin{subequations}
\begin{align}
&U_{\mathrm{T}}\phi_1U^\dagger_{\mathrm{T}}=\phi_2,~~U_{\mathrm{T}}\phi_2U^\dagger_{\mathrm{T}}=\phi_3,~~U_{\mathrm{T}}\phi_3U^\dagger_{\mathrm{T}}=\phi_1;\\
&U_{\mathrm{T}}\sigma_1U^\dagger_{\mathrm{T}}=\sigma_2,~~~U_{\mathrm{T}}\sigma_2U^\dagger_{\mathrm{T}}=\sigma_3,~~~U_{\mathrm{T}}\sigma_3U^\dagger_{\mathrm{T}}=\sigma_1.
\end{align}
\end{subequations}
This further allows us to further simplify the operator in Eq.(\ref{hesim1}) by choosing $\beta_{12}=\beta_{23}=\beta_{31}:=\beta$. Which this choice the success probability of discrimination the states $\chi^{\mathrm{E}}$ and $\chi^{\mathrm{O}}$, consequently the success probability of $\mathbb{P}[3]$ tasks in minimal tensor product theory of three qubits, become 
\begin{align}
P_{succ}=\tfrac{1}{2}\Big(1+\tfrac{9}{2}\beta\Big).
\end{align}
Notably, the success probability is independent of the parameter $\gamma$. Now, choosing $\gamma=0$ and $\beta=2/9$, the operator in Eq.(\ref{hesim1}) becomes
\begin{align}
\mathrm{h}'_{\mathrm{E}}&=\tfrac{1}{2}\Big(\mathbf{I}^{\otimes3}+\tfrac{2}{9}\Big([[\sigma_1,\sigma_2,\mathbf{I}]]+[[\sigma_2,\sigma_3,\mathbf{I}]]\nonumber\\
&\hspace{4cm}+[[\sigma_3,\sigma_1,\mathbf{I}]]\Big)\Big),\label{hesim2}
\end{align}    
ensuring the perfect success. What we are left with to argue that the operator in Eq.(\ref{hesim2}) yields positive probability on three-qubit product states. For that, consider a generic product state of the form
\begin{align}
\rho_{prod}:=\tfrac{1}{2}(\mathbf{I}+\hat{n}.\sigma)\otimes\tfrac{1}{2}(\mathbf{I}+\hat{m}.\sigma)\otimes\tfrac{1}{2}(\mathbf{I}+\hat{r}.\sigma),
\end{align}
where $\hat{n},~\hat{m},~\hat{r}$ be the unit vectors in $\mathbb{R}^3$. A straightforward calculation yields
\begin{align*}
\tr[\rho_{prod}\mathrm{h}'_\mathrm{E}]=\tfrac{1}{2}\Big(1+\tfrac{2\sqrt{3}}{9}(\hat{n}\times\hat{m}+\hat{m}\times\hat{r}+\hat{r}\times\hat{n}).\hat{\mathrm{T}}\Big).
\end{align*}
Denoting $S=\hat{n}\times\hat{m}+\hat{m}\times\hat{s}+\hat{s}\times\hat{n}$, the maximum value of $|S|^2$, under the constraint that the Gram determinant of $\hat{n}.\hat{m},~\hat{m}.\hat{s},~\hat{s}.\hat{n}$ has to be positive, is upper bounded by $27/4$ i.e.
\begin{align}
|S|^2\le\tfrac{27}{4}\implies -\tfrac{3\sqrt{3}}{2}\le|S|\le \tfrac{3\sqrt{3}}{2}.
\end{align}
Thus we have,
\begin{align*}
\tr[\rho_{prod}\mathrm{h}'_\mathrm{E}]\ge\tfrac{1}{2}\Big(1-\tfrac{2\sqrt{3}}{9}\tfrac{3\sqrt{3}}{2}\Big)=0.
\end{align*}
This completes the proof.
\end{proof}
\noindent Note that, although Proposition~\ref{prop3} employs a product-state preparation, the corresponding measurement is not a product one; rather, it belongs to a class of effects that are more general than entangled quantum measurements. This naturally leads to the question: can the $\mathbb{P}[3]$ game be perfectly accomplished using three systems of operational dimension $2$, without invoking entanglement either at the preparation or at the measurement stage? In the following we establish that this is indeed the case. For that, we begin by recalling a class of GPTs known as polygon models \cite{Janotta2011}, which have been extensively investigated in the literature \cite{Massar2014,Brunner2014,Banik2015,Banik2019,Bhattacharya2020,Saha2020}.

\subsection{$\mathbb{P}[3]$ in Polygon Models}\label{appenC}
{\it State Space:} For a given \(n\ge4\), state space $\Omega_n$ is the convex hull of $n$ pure states $\{\omega_i\}_{i=1}^{n}$ having the shape of a regular polygon, where
\begin{align}
\omega_i=\begin{pmatrix}
r_n\cos\left(\frac{2 \pi i}{n}\right)\\
r_n\sin\left(\frac{2 \pi i}{n}\right)\\
1
\end{pmatrix}\in\mathbb{R}^3,\label{polystate}
\end{align}
where $r_n:=\sqrt{\sec(\pi/n)}$.

{\it Effect Space:} The corresponding effect space $\mathcal{E}_n$ is the convex hull of $\{\theta,u,e_i,\bar{e}_i:=u-e_i\}_{i=1}^{n}$, with $\theta=(0,0,0)^{\mathrm{T}}$ and $u=(0,0,1)^{\mathrm{T}}$, and $e_i$'s are given by,
\begin{subequations}
\begin{align}
e_i&=\frac{1}{2}\begin{pmatrix}
r_n\cos\Big(\frac{(2i-1)\pi}{n}\Big)\\
r_n\sin\Big(\frac{(2i-1)\pi}{n}\Big)\\
1
\end{pmatrix}~~~\text{(for even)};\\
e_i&=\frac{1}{1+r^2_n}\begin{pmatrix}
r_n\cos\Big(\frac{2\pi i}{n}\Big)\\
r_n\sin\Big(\frac{2\pi i}{n}\Big)\\
1
\end{pmatrix}~~~\text{(for odd)}.
\end{align}    
\end{subequations}
Notably, for even $n$, $\bar{e}_i= e_{\big(i+\tfrac{n}{2}\big)\text{mod}~n}$.

{\it Reversible Transformations:} The set of the reversible transformations, $\mathcal{T}_n$, is the dihedral group of order $2n$ containing $n$ reflections and $n$ rotations,
\begin{align}
\mathcal{T}_n\equiv&\left\{T_k^p~|~k=0,\cdots,n-1;~\&~p\in\{+,-\} \right\},\nonumber\\
T_k^p&:= 
\begin{pmatrix}
\cos(\frac{2\pi k}{n}) & -p\sin(\frac{2\pi k}{n}) & 0 \\
\sin(\frac{2\pi k}{n}) & ~~p\cos(\frac{2\pi k}{n}) & 0 \\
0  & 0  & 1   
\end{pmatrix},
\end{align}
with $p=+$ corresponding to rotations and $p=-$ to  reflections.

\begin{proposition}\label{prop4}
The permutation parity task $\mathbb{P}[3]$ can be accomplished with unit success probability in the hexagon model, without invoking entanglement either in the preparation stage or in the measurement stage.
\end{proposition}
\begin{proof}
Let three hexagon-bits A,B,C be prepared in the product state $\omega_{1}\otimes\omega_{3}\otimes\omega_{5}$ (see Fig.~\ref{fig1}). The corresponding even and odd permutation states are
\begin{align*}
\eta^{\mathrm{E}}_1&:=\omega_1\otimes\omega_3\otimes\omega_5,\quad &\eta^{\mathrm{O}}_1&:=\omega_1\otimes\omega_5\otimes\omega_3,\\
\eta^{\mathrm{E}}_2&:=\omega_5\otimes\omega_1\otimes\omega_3,\quad &\eta^{\mathrm{O}}_2&:=\omega_3\otimes\omega_1\otimes\omega_5,\\
\eta^{\mathrm{E}}_3&:=\omega_3\otimes\omega_5\otimes\omega_1,\quad &\eta^{\mathrm{O}}_3&:=\omega_5\otimes\omega_3\otimes\omega_1.
\end{align*}
Consider the following set of product effects
\begin{align*}
E^1_{\mathrm{E}}&:=\tfrac{8}{9}~ e_2 \otimes e_4 \otimes e_6,\quad\quad E^1_{\mathrm{O}}:=\tfrac{8}{9}~ e_1 \otimes e_5 \otimes e_3;\\
E^2_{\mathrm{E}}&:=\tfrac{8}{9}~ e_4 \otimes e_6 \otimes e_2,\quad\quad E^2_{\mathrm{O}}:=\tfrac{8}{9}~ e_3 \otimes e_1 \otimes e_5,\\
E^3_{\mathrm{E}}&:=\tfrac{8}{9}~ e_6 \otimes e_2 \otimes e_4,\quad\quad E^3_{\mathrm{O}}:=\tfrac{8}{9}~ e_5 \otimes e_3 \otimes e_1,\\
&~~~~E^i_{\times}:=\tfrac{4}{9}~e_i\otimes e_i\otimes e_i~;\quad~i\in\{1,\cdots,6\}.
\end{align*}
It turns out that $\sum_{i=1}^6E^i_{\times}+\sum_{j=1}^3(E^j_{\mathrm{E}}+E^j_{\mathrm{O}})=u^{\otimes 3}$, thereby forming a 3-hexagonbit non-entangling measurement. Furthermore, $E^i_{\times}(\eta^{\mathrm{E}/\mathrm{O}}_j)=0,~\forall~i\in\{1,\cdots,6\}~\&~j\in\{1,2,3\}$; and $E^i_a(\eta^b_j)=\kappa_{ij}\delta_{ab}$ with $\kappa_{ij}\ge0,~\forall~i,j\in\{1,2,3\}$. Thus this measurement perfectly distinguishes the sets of odd parity states vs set of even parity states, thereby perfectly achieving the $\mathbb{P}[3]$ task. This completes the proof.
\end{proof}

\subsection{$\mathbb{P}[3]$ in Cube Theory}
While Proposition \ref{prop4} establishes that the $\mathbb{P}[3]$ task can be perfectly accomplished in hexagon theory without invoking entanglement in either preparation or measurement, the nature of the measurement warrants a closer examination. Although all of its effects are fully product, the measurement nonetheless requires the three hexagon bits to be addressed jointly; it cannot be realized by measuring each subsystem independently. In quantum theory, this peculiar feature is well known as nonlocality without entanglement \cite{Bennett1999}. More recently, analogous phenomena have been identified in the broader setting of GPTs \cite{Bhattacharya2020}. This naturally raises the question: can the $\mathbb{P}[3]$ task be achieved using GPT systems with outcome dimension $2$, while using product preparation and employing a measurement implementable through purely local addressing of each subsystem? Our next result answers this question affirmatively.
\begin{figure}[t!]
\centering
\includegraphics[width=1\linewidth]{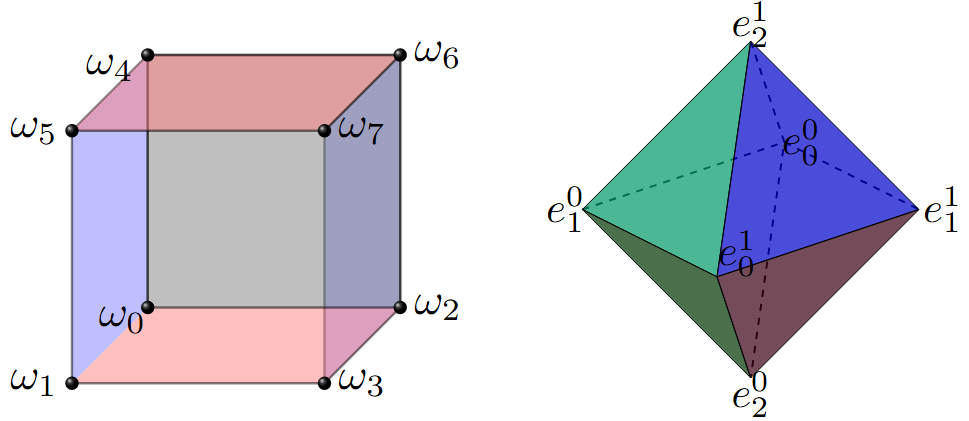}
\caption{Cube theory: the normalized state space  $\Omega_{\text{cu}}\equiv\mathrm{Conv.Hull}\{\omega_k
\}_{k=0}^7$. It has the six ray extreme effects $\{e^j_i\}_{i,j=0}^{2,1}$.  The state and effect cones live in $\mathbb{R}^4$.}\label{fig2}
\vspace{-.2cm}
\end{figure}

\begin{proposition}\label{prop5}
The permutation parity task $\mathbb{P}[3]$ can be accomplished perfectly in cube theory with product preparation and locally implementable product measurement.  
\end{proposition}
\begin{proof}
The normalized state space of cube theory is the convex hull of the eight pure states, i.e. $\Omega_{\text{cu}}\equiv\mathrm{Conv.Hull}\{\omega_k
\}_{k=0}^7$; and the normalized effect space is $\mathcal{E}_{\text{cu}}\equiv\mathrm{ConvHull}\{u, \Theta,e^j_i~|~i=0,1,2;~j=0,1\}$ (see Fig.~\ref{fig2}) of which $\{e^j_i\}$'s are the ray-extremal effect\footnote{For a cone $\mathrm{C}\subseteq\mathbb{R}^n$, a nonzero vector $\vec{r}\in\mathrm{C}$ is called a ray extremal element if there do not exists linearly independent vectors $\vec{r}_1,\vec{r}_2\in\mathrm{C}$ and positive scalars $\lambda_1,\lambda_2$ such that $\vec{r}=\lambda_1\vec{r}_1+\lambda_2\vec{r}_2$.}. The theory allows only three extremal measurements $\mathrm{M}_{i}\equiv\{e^0_i,e^1_i\}$ for $i=0,1,2$, implying outcome dimension of this model to be $2$. Denoting $\omega_k\equiv\omega_{x_2x_1x_0}$ with $k:=\sum_{i=0}^2x_i\times2^i$ and $x_0,x_1,x_2\in\{0,1\}$, it is immediate to see which ray-extreme effects click perfectly on which extreme states. In particular, the state $\omega_{x_2x_1x_0}$ gets clicked perfectly on the set of effects $\{e^{x_i}_i\}_{i=0}^2$. 

To accomplish the $\mathbb{P}[3]$ task is cube theory, consider three cube-bit systems $A, B, C$ prepared in the product state $\omega_{001}\otimes\omega_{010}\otimes\omega_{100}$. Under different permutations the state becomes 
\begin{align*}
\mathcal{C}^{\mathrm{E}}_1&:=\omega_{001}\otimes\omega_{010}\otimes\omega_{100},\quad &\mathcal{C}^{\mathrm{O}}_1&:=\omega_{001}\otimes\omega_{100}\otimes\omega_{010},\\
\mathcal{C}^{\mathrm{E}}_2&:=\omega_{010}\otimes\omega_{100}\otimes\omega_{001},\quad &\mathcal{C}^{\mathrm{O}}_2&:=\omega_{010}\otimes\omega_{001}\otimes\omega_{100},\\
\mathcal{C}^{\mathrm{E}}_3&:=\omega_{100}\otimes\omega_{001}\otimes\omega_{010},\quad &\mathcal{C}^{\mathrm{O}}_3&:=\omega_{100}\otimes\omega_{010}\otimes\omega_{001}.
\end{align*}
To identify the parity of the permutation, we perform the measurements $\mathrm{M}_0,~\mathrm{M}_1$, and $\mathrm{M}_2$ on first, second, and the third cube-bit systems, respectively.
On the state $\mathcal{C}^{\mathrm{E}}_1$ the measurements yield outcomes $(1,1,1)$. Similarly, on both the states $\mathcal{C}^{\mathrm{E}}_2$ and $\mathcal{C}^{\mathrm{E}}_3$ the outcomes are $(0,0,0)$. On the other hand, on states $\mathcal{C}^{\mathrm{O}}_1$, $\mathcal{C}^{\mathrm{O}}_2$, and $\mathcal{C}^{\mathrm{O}}_3$ the outcomes are respectively $(1,0,0),~(0,0,1),~(0,1,0)$. From this product measurement outcomes we thus perfectly distinguish the parity of the applied permutation, thereby completing the proof.  
\end{proof}
\section{Playing the $\mathbb{P}[4]$ game in Quantum Theory}
\noindent We now revisit the parity game in the four-party scenario. If one attempts to implement this game using a two-dimensional classical system, the success probability is upper bounded by $\tfrac{1}{2}$. In this section, we present a realization based on qubit product state preparation that surpasses this classical bound, achieving a success probability greater than $\tfrac{1}{2}$.

\noindent\textbf{Four qubit preparation}: Consider that four tetrahedron states
\begin{align}
\left.\begin{aligned}
\ket{\psi_1}& = \ket{0},\hspace{0.5 cm} \ket{\psi_2} = \frac{1}{\sqrt{3}}(\ket{0}+ \sqrt{2}\ket{1}),\\
\ket{\psi_3}&=\frac{1}{\sqrt{3}}(\ket{0}+ \sqrt{2} e^{\frac{2 \pi i}{3}}\ket{1}),\\
\ket{\psi_4}&=\frac{1}{\sqrt{3}}(\ket{0}+ \sqrt{2}e^{\frac{4 \pi i}{3}}\ket{1})\\
\end{aligned}\right.
\end{align}
Let $\ket{\psi} = \ket{\psi_1}\otimes\ket{\psi_2}\otimes\ket{\psi_3}\otimes\ket{\psi_4}$ and 
\begin{align}
\chi^{\mathrm{E}}:=\frac{1}{12}\sum_{\pi\in E_4}\pi(\ket{\psi}\bra{\psi}),~~\chi^{\mathrm{O}}:=\frac{1}{12}\sum_{\pi\in O_4} \pi(\ket{\psi}\bra{\psi}).    
\end{align}
The parity problem $\mathbb{P}[4]$ is now just a distinguishability problem of the set $\{\chi^{\mathrm{E}},\chi^{\mathrm{O}}\}$, where the optimal distinguishability comes from the Helstrom measurement. Using Helstrom measurement, the optimal succes of distinguishing the two states is
\begin{align}
\mathcal{P}^{[4]}_{\psi^\star_{p}}=\frac{1}{2}\left(1+\frac{2}{9}\right)=\frac{11}{18}> \frac{1}{2}.
\end{align}

\noindent As mentioned in Section II, a $n$-party parity game can be won perfectly using a$\lceil \sqrt{n} \rceil$-dimensional quantum system. Therefore, for $n > 4$, the game cannot be won perfectly using qubits alone. On the other hand, if one uses a two-dimensional classical system, the success probability is at most $\tfrac{1}{2}$. This naturally raises the question: if we play the game with $n > 4$ using qubits, can we obtain any advantage over the classical case? We address this question in the following theorem.

\section{$\mathbb{P}[n]$ game and the Linear Dimension}

\noindent While so far we have mostly analyzed the $\mathbb{P}[3]$ task, in this section we establish a generic result for having nontrivial success in $\mathbb{P}[n]$ in an arbitrary GPT.  

\begin{theorem}\label{theo1}
Let $\mathcal{S}$ be a GPT system with $\mathtt{Ld}_{\mathcal{S}}$.The optimal success probability in the parity game $\mathbb{P}[n]$ is strictly larger than $\frac{1}{2}$ if and only if $\mathtt{Ld}_{\mathcal{S}}\ge n$.  
\end{theorem}

\begin{proof}
Let $\omega \in \Omega_{\mathcal{S}^{\otimes n}}$, where $\otimes$ denotes any composition lying 
between the minimal and maximal tensor products. The symmetric group $S_n$ acts naturally on the tensor product space $\Omega_{\mathcal{S}^{\otimes n}}$ by permuting the subsystems. Let $E_n \subset S_n$ and $O_n \subset S_n$ denote the sets of even and odd 
permutations, respectively. The problem of determining parity of a permutation reduces to distinguishing the states
\begin{align*}
    \rho_e=\frac{2}{n!}\sum_{\pi\in E_n}\pi(\omega)
    \qquad \text{and} \qquad
    \rho_o=\frac{2}{n!}\sum_{\pi\in O_n}\pi(\omega).
\end{align*}
where $\pi(\omega)$ denotes the state obtained by permuting the $n$ subsystems according to $\pi$.

If $\rho_e=\rho_o$, the optimal success probability is $\frac{1}{2}$. We first show that this occurs whenever $n>\mathtt{Ld}_{\mathcal{S}}$. 

Let $\{v_i\}_{i=1}^{\mathtt{Ld}_{\mathcal{S}}}$ be a basis of the vector space $\mathbb{R}^{\mathtt{Ld}_{\mathcal{S}}}$. Any state $\omega \in \Omega_{\mathcal{S}^{\otimes n}}$ can be written as,
\begin{align*}
    \omega = \sum_{i_1,\dots,i_n=1}^{\mathtt{Ld}_{\mathcal{S}}} a_{i_1i_2\cdots i_n} v_{i_1}\otimes v_{i_2}\otimes\cdots v_{i_n}
\end{align*}
where $a_{i_1i_2\cdots i_n} \in \mathbb{R} ~~\forall~ i_j \in [\mathtt{Ld}_{\mathcal{S}}]$, $j\in [n] $.\\

Since $n> \mathtt{Ld}_{\mathcal{S}}$, the Pigeonhole principle implies that in every tensor product term there exist indices $k \neq l$ such that $i_k = i_l$. Consequently, exchanging the $k-$th and $l-$th subsystems leaves that term invariant. Hence each tensor product term is invariant under a transposition of two subsystems. Since composing any permutation with such a transposition flips its parity without changing its action on the tensor term, the contributions of even and odd permutations coincide. Therefore,
\begin{align}
    \rho_e=\rho_o\qquad \forall~ n > \mathtt{Ld}_{\mathcal{S}}
\end{align}
Hence the optimal success probability of the parity game equals $\frac{1}{2}$ in this case. \\

We now prove the converse. Suppose $n \leq \mathtt{Ld}_{\mathcal{S}}$.  Let $\mathrm{B}=\{v_i\}_{i=1}^{\mathtt{Ld}_{\mathcal{S}}}\subset \Omega_{\mathcal{S}}$ be a set of $\mathtt{Ld}_{\mathcal{S}}$ linearly independent states and consider the product state 
\begin{align}
   \omega = v_{i_1}\otimes v_{i_2}\otimes\cdots v_{i_n}\in \Omega_{\mathcal{S}^{\otimes n}}
\end{align}
with all indices $i_1,\cdots,i_n$ distinct. Since the vectors $\{v_i\}$ are linearly independent and $n \leq \mathtt{Ld}_{\mathcal{S}}$, the set of permuted states
\begin{align*}
    S_n(\omega)=\{\pi(\omega)|\pi\in S_n\}
\end{align*}is linearly independent \cite{}. Consequently, the subsets
\begin{align*}
    E_n(\omega)=\{\pi(\omega)|\pi\in E_n\}, \qquad O_n(\omega)=\{\pi(\omega)|\pi\in O_n\}
\end{align*}
span linearly independent subspaces.

By construction $\rho_e$ and $\rho_o$ are supported on spans of $E_n(\omega)$ and $O_n(\omega)$ respectively. Since these are linearly independent, the states $\rho_e$ and $\rho_o$ are distinct. Therefore, there exists a protocol such that the two states can be distinguished with probability strictly greater than $\frac{1}{2}$. This completes the proof.
\end{proof}

Theorem~\ref{theo1} establishes that a GPT exhibits a nontrivial advantage in the parity-permutation problem if and only if its linear dimension satisfies $n \leq \mathtt{Ld}_{\mathcal S}$. The proof is constructive, providing explicit product-state strategies that achieve a success probability strictly greater than $1/2$ whenever this condition is satisfied. Thus, the existence of a nontrivial advantage is completely characterized by the linear dimension of the underlying state space. Specializing to quantum theory, where $\mathtt{Ld}_{\mathcal S}=d^2$, with $d$ denoting the local Hilbert-space dimension, Theorem~\ref{theo1} implies that a quantum advantage is possible if and only if $n \leq d^2$.

Consequently, for systems composed of $n \geq 5$ qubits ($d=2$, $\mathtt{Ld}_{\mathcal S}=4$), no strategy can achieve a success probability exceeding random guessing, regardless of the amount of entanglement present in the preparation or measurement. 

This completely characterizes the threshold for the emergence of nontrivial parity information within the GPT framework.

\section{Discussion}
In \cite{Diebra2025}, the authors studied the permutation parity problem in the regime of perfect success using the minimum number of systems. They showed that, by employing entangled states, the parity of a permutation can be identified with certainty whenever $d \leq \lceil \sqrt{n} \rceil$, where (d) denotes the Hilbert space dimension of the local systems. Furthermore, they quantified the minimum amount of entanglement required to win the task with certainty.

In contrast, perfect success is not necessary for demonstrating a quantum advantage. Any success probability exceeding the optimal classical value already establishes an operational advantage of using non-classical systems. Motivated by this observation, we investigated the role of entanglement and the underlying state-space structure in the permutation parity problem $\mathbb{P}[n]$.

Our results show that the onset of an advantage over classical strategies is governed not by entanglement, but by the linear dimension of the underlying operational theory. Specifically, we prove that a generalized probabilistic theory (GPT) allows better identification of permutation parity than classical theory if and only if its linear dimension satisfies $n \leq \mathtt{Ld}_{\mathcal S}$. Thus, the existence of a nontrivial advantage is completely characterized by the linear dimension of the state space. Moreover, such an advantage can always be achieved using appropriately chosen product-state preparations, demonstrating that entanglement at the preparation stage is not necessary for outperforming classical strategies.

Within quantum theory, our results imply that the parity problem exhibits a clear separation between the resources required for obtaining an advantage and those required for perfect success. While entanglement can enable perfect parity identification, a quantum advantage already arises from suitably chosen product-state preparations. In particular, the proof of Theorem~\ref{theo1} provides a general construction yielding nontrivial advantage for all $n \leq \mathtt{Ld}_{\mathcal S}$.

To illustrate these findings explicitly, we constructed product-state qubit strategies for the cases $n=3$ and $n=4$, both of which surpass the corresponding classical success probabilities. These examples demonstrate that quantum advantage in the permutation parity problem can emerge solely from the larger linear dimension of the quantum state space, without requiring entangled preparations.

Our results therefore complement those of \cite{Diebra2025}. Whereas the latter identifies conditions for perfect parity discrimination in quantum theory, our work characterizes the threshold for the existence of nontrivial parity information in arbitrary GPTs. From this perspective, the permutation parity game serves as a dimension witness not only for quantum theory but for generalized probabilistic theories more broadly: the threshold $n \leq \mathtt{Ld}_{\mathcal S}$ precisely determines when parity becomes operationally detectable beyond the classical limit. 

Beyond their foundational significance, our results may also be relevant from an experimental perspective.Since product-state preparations are generally easier to realize and maintain than highly entangled states, the strategies presented here may provide a more accessible route for demonstrating quantum advantage in permutation parity tasks. This suggests that nontrivial quantum advantages in such games may be observed in near-term implementations without requiring the preparation of large entangled resources.\bigskip

\noindent{\bf Acknowledgement}: JK acknowledges support from University Grants Commission, India (Reference no. 231620167137). BC acknowledges support from University Grants Commission, India (Reference no. 241620129062). SRC acknowledges support from University Grants Commission, India (Reference no. 211610113404). SGN acknowledges support from the CSIR project $09/0575(15951)/2022$-EMR-I. KA acknowledges support from the CSIR project $09/0575(19300)/2024$-EMR-I. MB acknowledges the financial support through the National Quantum Mission (NQM) of the Department of Science and Technology, Government of India.

\bibliography{reference}

\end{document}